\documentclass[preprint,12pt]{article}
\usepackage[T1]{fontenc}
\usepackage[utf8]{inputenc}
\usepackage{amsmath,amsfonts,amssymb,amsthm}
\usepackage{tikz}
\usepackage{ulem}
\usepackage{todonotes}
\usepackage{framed}
\usepackage{xspace}

\usetikzlibrary{calc,positioning}
\renewcommand{\emph}[1]{{\textit{ #1}}}

\newtheorem{theorem}{Theorem}

\newtheorem{clm}{Claim}
\newtheorem{col}[theorem]{Corollary}
\newtheorem{prop}[theorem]{Proposition}
\theoremstyle{definition}
\newtheorem{problem}{Problem}

\newenvironment{inproof}{{\it Proof.}}{\hfill$\blacksquare$ \medskip}
\newcommand{\kradius}[1]{\textsc{$#1$-Radius Sequence}\xspace}
\newcommand{\kcover}[1]{\textsc{$#1$-Cover Sequence}\xspace}
\newcommand{\dist}{\operatorname{dist}}
\newcommand{\mc}{\operatorname{mc}}
\begin{document}

\title{Sequences of radius $k$ for complete bipartite graphs\footnote{Research supported by the Polish National Science Center, decision nr DEC-2012/05/B/ST1/00652. The third author was partially supported by ERC Starting Grant PARAMTIGHT (No. 280152). The extended abstract of this paper was presented on the conference WG 2016 \cite{DLRz-WG}}}

\author{Micha{\l} D\k{e}bski
\footnote{Faculty of Mathematics and Information Science, Warsaw University of Technology, Warsaw, Poland}
\and 
Zbigniew Lonc
\footnotemark[2]
\and
Pawe{\l} Rz\k{a}\.{z}ewski \footnotemark[2]
\footnote{Institute of Computer Science and Control, Hungarian Academy of Sciences (MTA SZTAKI), Budapest, Hungary}
}
\date{}

\maketitle
\begin{abstract}
A \emph{$k$-radius sequence} for a graph $G$ is a sequence of vertices of $G$ (typically with repetitions) such that for every edge $uv$ of $G$ vertices $u$ and $v$ appear at least once within distance $k$ in the sequence. The length of a shortest $k$-radius sequence for $G$ is denoted by $f_k(G)$.
We give an asymptotically tight estimation on $f_k(G)$ for complete bipartite graphs {which matches a lower bound, valid for all bipartite graphs}. 
We also show that determining $f_k(G)$ for an arbitrary graph $G$ is NP-hard for every constant $k>1$.
\end{abstract}


\section{Introduction}

\subsection{$k$-radius sequences}
\label{subsection_kradiussequences}
Suppose we need to compute values of a two-argument function, say $H$, for all pairs of large objects. {In some applications,
the large size of objects mandate that most of those objects must be stored in the
secondary memory.} To compute the values of the function, we need to place these objects in our cache before carrying out the computations. {A good caching order reduces the number of accesses to the
secondary memory.} The cache is limited in size -- it can hold up to $k + 1$ objects at one time. Our task is to provide a shortest possible sequence of (costly) read operations to ensure that each pair of objects will at some point reside in the cache together so that we can compute the values of $H$ for all pairs of objects. This problem appeared in practice in processing large medical images (see Jaromczyk and Lonc \cite{JL}).

The read operation assumes that, if the cache is full, the next object takes the place of one of the objects currently residing in the cache. So far most of the research related to this problem has been concentrated on a special case when we assume that the replacement of objects is based on the first-in first-out strategy. This leads to the concept of a $k$-radius sequence. Let $k$ and $n$ be positive integers and let $V$ be an $n$-element set (of objects). We say that a sequence (with possible repetitions) of elements of $V$ is a {\it $k$-radius sequence} (or has a {\it $k$-radius property}) if every two elements in $V$  are at distance at most $k$ somewhere in the sequence. Observe that short $k$-radius sequences correspond to efficient caching strategies for our problem. Indeed, if $x_1,x_2,\ldots,x_m$ is a $k$-radius sequence, then at time $t$ we load the element $x_t$ and after this loading (for $t\geq k+1$) the cache holds the elements $x_{t-k},x_{t-k+1},\ldots,x_t$. The $k$-radius property guarantees that any pair of elements of $V$ resides in the cache together at some point. We denote by $f_k(n)$ the length of a shortest $k$-radius sequence over an $n$-element set of objects.

The problem of constructing short $k$-radius sequences has been considered by several researchers (see Blackburn \cite{Bl}, Blackburn and McKee \cite{BM}, Chee {\it et al.}~\cite{CLTZ}, D\k ebski and Lonc \cite{DL}, Jaromczyk and Lonc \cite{JL}, Jaromczyk {\it et al.} \cite{JLT}, {Bondy {\it et al.} \cite{BLRz}}).

\subsection{$k$-radius sequences for graphs}

In this paper we consider a more general problem -- we assume that the values of the function $H$ need not be computed for all pairs of objects but only for some of them. Let $V$ be a set of objects and let $G=(V,E)$ be a graph. We ask: what is the smallest number $c_k(G)$ of read operations that guarantees that each pair of vertices adjacent in $G$ resides in the cache together at some point? {We refer the reader to the second paragraph of Section \ref{complexity} for a precise definition of the parameter $c_k(G)$.}

If we assume additionally that the replacement of objects in the cache is based on the first-in first-out strategy, then we get the following generalization of $k$-radius sequences. A sequence (typically with repetitions) of vertices of a graph $G=(V,E)$ is called a {\it $k$-radius sequence for $G$} (or alternatively, it has a {\it $k$-radius property with respect to $G$}) if each pair of adjacent vertices of $G$ appears at distance at most $k$ in this  sequence. More precisely, a sequence $x_1,x_2,\ldots,x_m$, $x_i\in V$, of vertices of a graph $G=(V,E)$ is called a $k$-radius sequence if for each two vertices $u$ and $v$ adjacent in $G$ there are $i$, $j$, $1\leq i,j\leq m$, such that $u=x_i$, $v=x_j$ and $|j-i|\leq k$. We denote by $f_k(G)$ the length of a shortest $k$-radius sequence for the graph $G$. Clearly, assuming the first-in first-out strategy, $f_k(G)$ is equal to the least number of read operations that guarantees that each pair of vertices adjacent in $G$ resides in the cache together at some point. Thus $f_k(G) \geq c_k(G)$.

We will always assume that $G$ has more than $k+1$ non-isolated vertices; otherwise finding $c_k(G)$ and $f_k(G)$ is trivial. If $G$ satisfies this condition, then there is an obvious lower bound for both numbers $c_k(G)$ and $f_k(G)$:
\begin{equation}\label{basic_lower_bound}
f_k(G)\geq c_k(G)\geq \frac{e(G)}{k}+\frac{k+1}{2},
\end{equation}
where $e(G)$ is the number of edges in $G$.

Indeed, consider a strategy that requires $m=c_k(G)$ read operations only and guarantees that each pair of vertices resides in the cache together at some point. Observe that if after loading a vertex the cache stores $j$ vertices, then it contains at most $j-1$ pairs of adjacent vertices which were not together in the cache before. Thus, as we start from an empty cache, after $m$ read operations at most $0+1+\ldots+(k-1)+(m-k)k=mk-{k+1\choose 2}$ pairs of adjacent vertices have been in the cache together at some point. Consequently, $e(G)\leq mk-{k+1\choose 2}$, which is equivalent to (\ref{basic_lower_bound}).


The original $k$-radius sequence problem is a special case of our generalization, where $G=K_n$ (the complete graph on $n$ vertices). Blackburn \cite{Bl} gave a simple replacement strategy {which} shows that, for a fixed $k$, {the value of} $c_k(K_n)$ is asymptotically equal to the lower bound (\ref{basic_lower_bound}). Moreover, he proved using a non-constructive method that imposing the restriction to a first-in first-out strategy does not affect the asymptotic efficiency, i.e. the number $f_k(K_n)$ is asymptotically equal to the lower bound (\ref{basic_lower_bound}) too. Currently, the best known upper bound for $f_k(K_n)$ is $f_k(K_n)=\frac{n^2}{2k}+O(n \log n)$, which was proved by a constructive method in our recent paper~\cite{DLRz-cyrkulanty}.

Now, consider the case when $G$ is a complete bipartite graph $K_{m,n}$. In terms of the initial motivation it means that we want to compute the values of a two-argument function $H$ whose domain is a Cartesian product $X\times Y$, where $X$ and $Y$ are the sets that form the bipartition in $G$.

If $k$ is fixed and both $m$ and $n$ are large, then $c_k(K_{m,n})$ is asymptotically equal to the lower bound (\ref{basic_lower_bound}) -- more precisely, we have $c_k(K_{m,n})=\frac{mn}{k}+O(m+n)$. This bound is attained by the following replacement strategy: pick $k$ vertices from $X$ and keep them in the cache while cycling through all vertices from $Y$, and repeat the process a total of $\left\lceil\frac{\left|X\right|}{k}\right\rceil$ times, each time picking $k$ -- or possibly less than $k$ in the last iteration -- different vertices from $X$.

{It is perhaps interesting that unlike in the case of the complete graph $K_n$, the parameters $f_k(K_{m,n})$ and $c_k(K_{m,n})$ are not asymptotically equal (see Theorem \ref{theorem_KmnUpper}).}

\subsection{Our contributions}\label{sec:results}

The main result of this paper is that for every $k$ there is a constant $d_k$ such that $f_k(K_{m,n})$ is roughly equal to $d_k\frac{mn}{k}$, in case when $m$ and $n$ are sufficiently large -- that is, a shortest $k$-radius sequence for a complete bipartite graph is roughly $d_k$ times longer than the trivial lower bound (\ref{basic_lower_bound}) would imply. We have $1\leq d_k < 1+\frac{\sqrt{2}}{2} \approx 1.7071$ (and $d_k$ is close to $1+\frac{\sqrt{2}}{2}$ for large $k$). Here is a precise statement of this result (see the end of Section \ref{section_proofs} for the proof).

\begin{theorem}
\label{theorem_KmnUpper}
Let $k$ be a positive integer. For every $\epsilon>0$ if $m$ and $n$ are sufficiently large, then
$$d_k\frac{mn}{k} \leq f_k\left(K_{m,n}\right)\leq\left(1+\epsilon\right)d_k\frac{mn}{k},$$
where $\frac{k}{2k-\sqrt{2k(k-1)}}\leq d_k\leq \frac{k+1}{2k-\sqrt{2k(k-1)}}.$
\end{theorem}

It is worth highlighting that one part of Theorem \ref{theorem_KmnUpper}, the lower bound, generalizes to all bipartite graphs. The following result is a reformulation of Corollary \ref{cor1}.

\begin{theorem}
\label{theorem_bipartiteLower}
Let $k$ be a positive integer. For every bipartite graph $G$ we have
$$
f_k(G)\geq d_k\frac{e(G)}{k},
$$
where $d_k$ is the constant from Theorem \ref{theorem_KmnUpper}.
\end{theorem}

\subsection{Related problems}

An additional motivation of our study comes from its relationship to maximum cuts in some graphs. A {\it maximum cut} in a graph $G$ is a bipartition of the set of vertices of $G$ maximizing the {\it size} of the cut, i.e. the number of edges that join vertices of the two sets of the bipartition; the size of the maximum cut in $G$ is denoted $\mc(G)$. Finding a maximum cut in a graph is a widely studied problem which is important in both graph theory and combinatorial optimization (see Newman \cite{N} and a survey by Poljak and Tuza \cite{PT}).

Let $C_n^k$ denote a circulant graph obtained from the cycle $C_n$ on $n$ vertices by joining with edges all vertices at distance at most $k$. Our considerations yield an estimation on the size of a maximum cut in $C_n^k$ (see the end of Section \ref{section_proofs} for a proof).

\begin{col}
\label{corollary_mcincirculant}
For a fixed $k$, we have
$$
\mc(C_n^k)=\frac{kn}{d_k}\left(1-o(1)\right),
$$
where $d_k$ is the constant from Theorem \ref{theorem_KmnUpper}.
\end{col}

The implications go both ways -- given the size of a maximum cut in a graph $G$ we can derive a lower bound on $f_k(G)$. Note that a $k$-radius sequence for $G$ must be also a $k$-radius sequence for every subgraph of $G$ (in particular, the bipartite subgraph induced by the maximum cut). With Theorem \ref{theorem_bipartiteLower}, it implies that:

\begin{col}
For every graph $G$, we have
$$f_k(G)\geq d_k \frac{\mc(G)}{k},$$
where $d_k$ is the constant from Theorem \ref{theorem_KmnUpper}.
\end{col}

The problem of finding a shortest $k$-radius sequence for a graph is also related to the {\it bandwidth problem}.
The {\it bandwidth} of a graph $G=(V,E)$ is the minimum of the values $\max \{|i-j| \colon v_iv_j \in E\}$ over all orderings $(v_1,v_2,\ldots,v_n)$ of $V$. Let us call such an ordering {\it bw-optimal}.
Informally speaking, we want to place the vertices of $G$ in integer points of a line in such a way that the longest edge is as short as possible (see for example Chinn {\it et al.} \cite{CCDG}).

Consider a graph $G$ with bandwidth $k$ and the bw-optimal ordering of its vertices. It is easy to observe that it is a $k$-radius sequence for $G$. Thus the graph (with no isolated vertices) has a $k$-radius sequence containing each vertex exactly once if and only if its bandwidth is at most $k$.
Since determining the bandwidth is NP-hard, even for subcubic graphs (see Garey {\it et al.} \cite{GGJK}), the problem of determining the existence of a $k$-radius sequence of length $n$ is NP-hard as well (if $k$ is a part of the instance).

In Section \ref{complexity} we give stronger complexity results. We show the problem of determining $f_k(G)$ for an arbitrary graph $G$ is NP-hard even if $k$ is a constant greater than $1$. Moreover, determining $c_k(G)$ for an arbitrary graph $G$ is NP-hard for every constant $k\geq 1$.


\section{Asymptotically shortest $k$-radius sequences for complete bipartite graphs}
\label{section_proofs}


For technical reasons it will be convenient to {consider} in this section  binary sequences {that} are {\it cyclic}. {The} terms of {such sequences} $b_1b_2\ldots b_s$ are arranged in a ``cyclic way'', i.e. $b_1$ is a successor of $b_s$. Consequently, we redefine the notion of the distance for cyclic sequences to ${\dist_c}(b_i,b_j)=\min(|i-j|,s-|i-j|)$. {The ``cyclic'' version of our problem is much more symmetric and greatly simplifies many of the following arguments.} {By a  \emph{cyclic $k$-radius sequence} for a graph $G$ we mean a sequence of vertices of $G$ such that for every edge $uv$ of $G$ vertices $u$ and $v$ appear at least once within cyclic distance $k$ in the sequence. We define $f_k^{cyc}(G)$ to be the length of a shortest cyclic $k$-radius sequence for a graph $G$. Obviously, every $k$-radius sequence is a cyclic $k$-radius sequence and if $x_1,x_2,\ldots,x_s$ is a cyclic $k$-radius sequence, then $x_1,x_2,\ldots,x_s,x_1,x_2,\ldots,x_k$ is a $k$-radius sequence. Thus, $f_k(G)-k\leq f_k^{cyc}(G)\leq f_k(G)$.}

When we construct a cyclic $k$-radius sequence for a bipartite graph $G$, we have to jump from one bipartition class of vertices to the other many times. Let $X$ and $Y$ be the bipartition classes in $G$, $|X|=m$ and $|Y|=n$ and let $\boldsymbol{ a}=a_1,a_2,\ldots,a_s$ be a cyclic $k$-radius sequence for the graph $G$. We define the binary sequence $\boldsymbol{b}(\boldsymbol{a})=b_1b_2\ldots b_s$ {(called a {\it characteristic sequence} of $\boldsymbol{a}$)} such that $b_i=0$ whenever $a_i\in X$ and $b_i=1$ whenever $a_i\in Y$. Every appearance of two identical symbols at cyclic distance at most $k$ in $\boldsymbol{b}(\boldsymbol{a})$ corresponds to a pair of vertices of $G$ which are at the same distance in $\boldsymbol{a}$ but do not form an edge in $G$. Therefore, we call the pair of indices of such a pair of terms in $\boldsymbol{b}(\boldsymbol{a})$ a {\it bad pair}.

Formally, an unordered pair $ij$, $i\not=j$, is a {\it $k$-bad pair} (resp. a {\it $k$-good pair}) in a cyclic binary sequence $\boldsymbol{b}$, if ${\dist_c}(b_i,b_j)=\min(|i-j|,s-|i-j|)\leq k$ and $b_i=b_j$ (resp. $b_i\not=b_j$). For every $k$ and $s$ we will be interested in constructing a cyclic binary sequence $\boldsymbol{b}$ of length $s$ with the least possible number of $k$-bad pairs. Let $w_k(s)$ be this  number.

The number of all pairs of terms at cyclic distance at most $k$ in a cyclic binary sequence of length $s$ is equal to $ks$. Let $M$ be the length of a shortest { cyclic} $k$-radius sequence for a bipartite graph $G$. Then, we get the inequality
\[
kM \geq e(G)  + w_k(M).
\]
So, if we prove that $w_k(s)\geq \alpha s$, for some $\alpha<k$, then we will get
\begin{equation}\label{inequ1}
{f_k(G) \geq} f_k^{{cyc}}(G) \geq \frac{e(G)}{k-\alpha}.
\end{equation}

Clearly, $w_1(s)=0$ if $s$ is even because the cyclic sequence $0101\ldots 01$ has no $1$-bad pairs. For a similar reason $w_1(s)=1$ when $s$ is odd.

Let $B_k$ be de Bruijn graph, i.e. a directed graph, whose vertices are all $k$-term binary sequences and an ordered pair of vertices $(v,u)$ is an edge if the $(k-1)$-term suffix of $v$ is the $(k-1)$-term prefix of $u$. We identify each edge with the $(k+1)$-term binary sequence which starts with the first term of $v$ and is followed by all the terms of $u$.

Clearly, every cyclic binary sequence of length $s$ corresponds to a directed closed walk of length $s$ in $B_k$ (both vertices and edges can appear in a walk an arbitrary number of times). We assign to every edge $e$ in $B_k$ the weight $t_k(e)$ which is equal to the number of appearances of the first term of $e$ on the remaining $k$ positions of $e$. For instance, if $e=010001$ (here $k=5$), then $t_5(e)=3$. The weight $t_k(C)$ of a closed walk $C$ in $B_k$ is just the sum of weights of its edges (we count each edge as many times as it appears in the walk).
\begin{prop}\label{prop1}
The number of $k$-bad pairs in a cyclic binary sequence is equal to the weight of the corresponding closed walk in de Bruijn graph $B_k$.
\end{prop}
\begin{proof}
To see this, it suffices to observe that every $k$-bad pair contributes to the weight of exactly one edge of the corresponding closed walk -- the edge starting with the element of the pair, which appears first in the sequence.
\end{proof}
The {\it normalized weight} of a closed walk $C$ in $B_k$ is the ratio $\frac{t_k(C)}{|C|}$ ($|C|$ is the number of edges in $C$ - again we count each edge as many times as it appears in $C$).

Let $a_k$ be the least possible normalized weight of a cycle in $B_k$, i.e.
\[
a_k=\min\left\{\frac{t_k(C)}{|C|}:\ C\ {\rm is\ a\ cycle\ in\ } B_k\right\}
\]
(we allow no multiple appearances of vertices and edges in cycles).
\begin{prop}\label{norm}
For all positive integers $k$ and $s$,
\[
a_ks\leq w_k(s) < a_ks+k(2^k+k).
\]
\end{prop}
\begin{proof} By Proposition \ref{prop1}, $w_k(s)$ is equal to the least possible weight of a closed walk, say $C$, of length $s$ in de Bruijn graph $B_k$. Clearly, the multiset of edges of the closed walk $C$ can be split into sets of edges of cycles, say $C_1,C_2,\ldots,C_p$, in $B_k$.

By the definition of $a_k$, we have $t_k(C_i)\geq a_k|C_i|$, for $i=1,\ldots,p$. Hence,
\[
w_k(s)=t_k(C)=t_k(C_1)+\ldots+t_k(C_p) \geq a_k(|C_1|+\ldots+|C_p|) = a_k|C|=a_ks.
\]

To complete the proof we need to construct a {cyclic} binary sequence of length $s$ with less than $a_ks+k(2^k+k)$ bad pairs. Let $\ell$ be the length of a cycle $C$ in $B_k$ with the normalized weight equal to $a_k$. Moreover, let $q=\lfloor\frac{s}{\ell}\rfloor$ and $r=s-q\ell\leq \ell-1<|V(B_k)| = 2^k$. We define $C'$ to be the closed walk in $B_k$ obtained by traversing the cycle $C$ $q$ times. Clearly, $t_k(C')=qt_k(C)=q\ell a_k \leq sa_k$.

We insert anywhere in the cyclic sequence corresponding to the closed walk $C'$ a sequence of $r$ consecutive $0$'s. The number of bad pairs in the resulting {cyclic} binary sequence is not larger than $t_k(C')+(k+r)k < a_ks+k(2^k+k)$.
\end{proof}
It follows from the proof of Proposition \ref{norm} that if $s$ is divisible by the length $\ell$ of a cycle in $B_k$ of minimum normalized weight (equal to $a_k$), then $w_k(s)=a_ks$ and there is a cyclic binary sequence with exactly $a_ks$ bad pairs which is periodic with the period equal to $\ell$.

Moreover, by Proposition \ref{norm}, {we have}
\begin{equation}\label{limit}
\lim_{s\rightarrow\infty}\frac{w_k(s)}{s}=a_k.
\end{equation}

Clearly, the cyclic sequence $0101\ldots 01$ of length $2s$ proves that $w_k(2s)\leq ks$, so $a_k\leq\frac{k}{2}<k$. Thus, Proposition \ref{norm} and the inequality (\ref{inequ1}) give now the following statement.
\begin{col}\label{cor1}
Let $k$ be a positive integer. For every bipartite graph $G$,
\[
{f_k(G) \geq} f_k^{{cyc}}(G) \geq \frac{e(G)}{k-a_k}.
\]
\end{col}\hfill$\Box$

We shall show now that this lower bound is asymptotically tight for $K_{m,n}$.

\begin{theorem}\label{upper}
For every integer $k$ and real $\varepsilon>0$, if $m$ and $n$ are sufficiently large, then
\[
f_k(K_{m,n}) \leq \frac{mn}{k-a_k}(1+\varepsilon).
\]
\end{theorem}
\begin{proof} We shall use the following theorem by Frankl and R\"odl \cite{FranklRodl} from hypergraph theory. {Here we allow multiple appearances of edges of a hypergraph.} Recall that a hypergraph is {\it $r$-uniform} if all its edges have cardinality $r$. It is {\it $d$-regular} if each of its vertices is contained in exactly $d$ edges. By the {\it codegree} ${\rm codeg}_H(v,u)$ of a pair of distinct vertices $v$ and $u$ in a hypergraph $H$ we mean the number of edges containing both $v$ and $u$. Finally, a {\it covering} of $H$ is a set of edges whose union is equal to the set of all vertices of $H$.

\medskip
\begin{minipage}{.95\linewidth}
{\it {\bf Theorem (Frankl, R\"odl \cite{FranklRodl})\footnotemark.} Let $r \in \mathbb{N}$ and $\delta > 0$ be fixed. There exist $d_0 \in \mathbb{N}$ and $\delta' >0$ such that for every $N \geq d \geq d_0$ the
following holds.
If $H$ is an $r$-uniform hypergraph with $N$ vertices satisfying the conditions:
\begin{enumerate}
\item $H$ is $d$-regular,
\item ${\rm codeg}_H(v,u) \leq \delta' \cdot d$ for any vertices $v,u$, $v\not=u$,
\end{enumerate}
then $H$ has a covering by at most $(1+\delta)\frac{N}{r}$ edges.
}
\end{minipage}
\footnotetext{This is a special case of a version of the original theorem that appears in Alon and Spencer \cite[Theorem 4.7.1]{AS}.}

\medskip
Let $C_k$ be a cycle in $B_k$ with the normalized weight equal to $a_k$. We denote by $\ell$ the length of $C_k$. Let $qC_k$ be the closed walk in $B_k$ obtained by traversing the cycle $C_k$ $q$ times, where $q=\left\lceil \frac{1+\varepsilon}{\varepsilon}\cdot\frac{k(k+1)}{\ell(k-a_k)}\right\rceil$. Clearly, the number of $k$-good pairs in the cyclic sequence $\boldsymbol{c}'$ corresponding to the closed walk $qC_k$ is $(k-a_k)q\ell$. Let $\boldsymbol{c}$ be the (non-cyclic) binary sequence of length $q\ell$ obtained from $\boldsymbol{c}'$ by cutting it at some point and let $r$ be the number of $k$-good pairs in $\boldsymbol{c}$.
{Observe that a $k$-good pair in $\boldsymbol{c}'$ is either still $k$-good in $\boldsymbol{c}$, or it is no longer at distance at most $k$.
{Since} {$\boldsymbol{c}$} has $\frac{k(k+1)}{2}$ {fewer} pairs at distance at most $k$ than {$\boldsymbol{c}'$}, we get $r\geq (k-a_k)q\ell-\frac{k(k+1)}{2}$.}
We denote by $c_0$ the number of $0$'s and by $c_1=q\ell - c_0$ the number of $1$'s in $\boldsymbol{c}$ . {Also, let $\delta'$ be {the value of the parameter in the Frankl-R\"{o}dl Theorem} for $\delta = \epsilon/2$ and {$r$ defined above.}}

Let $X$ and $Y$ be the bipartition classes in $K_{m,n}$, with $|X|=m$ and $|Y|=n$. We denote by $H$ the  hypergraph whose vertices are all ordered pairs $xy$ such that $x\in X$, $y\in Y$. For every sequence $\boldsymbol{a}$ of $q\ell$ distinct vertices in $K_{m,n}$ such that $\boldsymbol{b}(\boldsymbol{a})=\boldsymbol{c}$ we define an edge $e_{\boldsymbol{a}}$ in $H$. The edge $e_{\boldsymbol{a}}$ consists of such vertices $xy$ of $H$ that $x$ and $y$ are at distance at most $k$ in $\boldsymbol{a}$ {(see Table \ref{fig:exampleH} {for an example illustrating the construction of the hypergraph $H$})}. Clearly, there are $r$ such pairs $xy$ for every such $\boldsymbol{a}$, so $|e_{\boldsymbol{a}}|=r$ and, consequently, the hypergraph $H$ is $r$-uniform.

\begin{table}
\centering
\begin{tabular}{|l | l|}
\hline
sequence $\boldsymbol{a}$ & corresponding edge $e_{\boldsymbol{a}}$\\ \hline
$x_1,x_2,y_1,x_3,y_2$ & $\{x_1y_1,x_2y_1,x_3y_1,x_3y_2\}$ \\
$x_1,x_2,y_1,x_3,y_3$ & $\{x_1y_1,x_2y_1,x_3y_1,x_3y_3\}$ \\
$x_2,x_1,y_2,x_3,y_1$ & $\{x_2y_2,x_1y_2,x_3y_2,x_3y_1\}$ \\
{$x_1,x_2,y_2,x_3,y_1$} & $\{x_1y_2,x_2y_2,x_3y_2,x_3y_1\}$ \\
$x_1,x_2,y_2,x_3,y_3$ & $\{x_1y_2,x_2y_2,x_3y_2,x_3y_3\}$ \\
\multicolumn{1}{|c|}{$\vdots$} & \multicolumn{1}{|c|}{$\vdots$} \\
\end{tabular}
\caption{The construction of the hypergraph $H$ for $X=\{x_1,x_2,x_3\}$, $Y=\{y_1,y_2,y_3\}$, $\boldsymbol{c}=00101$, and $k=2$. Observe that the edge $\{x_1y_2,x_2y_2,x_3y_2,x_3y_1\}$ appears {a} multiple {number of} times.}
\label{fig:exampleH}
\end{table}

We observe that for every pair $xy$ there are exactly $d=r(m-1)\ldots(m-c_0+1)(n-1)\ldots(n-c_1+1)=\Theta(m^{c_0-1}n^{c_1-1})$ sequences $\boldsymbol{a}$ in which the vertices $x$ and $y$ are at distance at most $k$. Thus, the hypergraph $H$ is $d$-regular.

To estimate the maximum codegree in $H$ we consider three cases. Let us assume first that $u$ and $v$ are vertices in $H$ such that $u=xy_1,v=xy_2$, for some $x\in X$, $y_1,y_2\in Y$, where $y_1\not=y_2$. There are at most $r^2(m-1)\ldots(m-c_0+1)(n-1)\ldots(n-c_1+2)=\Theta(m^{c_0-1}n^{c_1-2})$ sequences $\boldsymbol{a}$ in which the vertex $x$ is at distance at most $k$ from both $y_1$ and $y_2$. Thus, the codegree ${\rm codeg}_{H}(u,v)=O(m^{c_0-1}n^{c_1-2})$.

If $u=x_1y$ and $v=x_2y$, where $y\in Y$, $x_1,x_2\in X$, $x_1\not=x_2$, then  we prove similarly that ${\rm codeg}_{H}(u,v)\leq r^2(m-1)\ldots(m-c_0+2)(n-1)\ldots(n-c_1+1)=O(m^{c_0-2}n^{c_1-1})$. Finally, let $u=x_1y_1,v= x_2y_2$, where $x_1,x_2\in X$, $y_1,y_2\in Y$, $x_1\not=x_2$, $y_1\not=y_2$. Then,  ${\rm codeg}_{H}(u,v)\leq r^2(m-1)\ldots(m-c_0+2)(n-1)\ldots(n-c_1+2)=O(m^{c_0-2}n^{c_1-2})$. In all the cases,  {we obtain that} ${\rm codeg}_{H}(u,v)<\delta'\cdot d$ for sufficiently large $m$ and $n$.

Thus, the assumptions of the Frankl-R\"odl Theorem are satisfied, so there is a covering of the vertex set of $H$ by at most ${(1+\delta)\frac{mn}{r}=}(1+\frac{\varepsilon}{2})\frac{mn}{r}$ edges. Let us consider a sequence obtained by concatenation of the sequences corresponding to these edges. In this sequence every two vertices forming an edge in $K_{m,n}$ are at distance at most $k$. The length of this sequence is at most
\begin{align*}
&\left(1+\frac{\varepsilon}{2}\right)\frac{mn}{r}q\ell \leq \left(1+\frac{\varepsilon}{2}\right)\frac{mn}{(k-a_k)q\ell-\frac{k(k+1)}{2}}q\ell\\
&= \left(1+\frac{\varepsilon}{2}\right)\frac{mn}{(k-a_k)\left(1-\frac{k(k+1)}{2q\ell(k-a_k)}\right)} \leq \frac{mn}{(k-a_k)}(1+\varepsilon),
\end{align*}
which completes the proof our theorem.
\end{proof}

In view of Theorem \ref{upper} and Corollary \ref{cor1} it would be interesting to find the exact values of $a_k$. The values of $a_k$ for $k\leq 5$ as well as the optimal cycles in $B_k$ (i.e. the cycles for which the normalized weight is equal to $a_k$) are shown in Table \ref{tabelka}. We denote by $(b_1b_2\ldots b_p)^*$ the cycle in $B_k$ whose consecutive edges are $b_1b_2\ldots b_k$, $b_2b_3\ldots b_{k+1}$, ... , $b_pb_1\ldots b_{k-1}$.

\begin{table}[h]
\centering
\begin{tabular}{|l | l | l|}
\hline
$k$ & $a_k$ & optimal cycles in $B_k$ \\ \hline
1 & 0 & $(01)^*$ \\
2 & 1/2 & $(0011)^*$ \\
3 & 1 & $(01)^*, (0011)^*, (000111)^*, (00011)^*, (00111)^*$ \\
4 & 4/3 & $(000111)^*$ \\
5 & 7/4 & $(00001111)^*$ \\ \hline
\end{tabular}
\caption{The values of $a_k$ and optimal cycles in $B_k$ for small $k$.}
\label{tabelka}
\end{table}

It is routine to show that the normalized weight of the cycle $(0^t1^t)^*$ ($t$ $0$'s followed by $t$ $1$'s) in $B_k$ is equal to $\frac{{t\choose 2}+{k-t+1\choose 2}}{t}$, for $\frac{k}{2}\leq t \leq k+1$. Let
\begin{equation}\label{def b_k}
{z_k}=\min_{\frac{k}{2}\leq t\leq k+1}\frac{{t\choose 2}+{k-t+1\choose 2}}{t}.
\end{equation}
Obviously, $a_k\leq {z_k}$. We conjecture that $a_k={z_k}$ for all positive integers $k$. The values of $a_k$ in Table \ref{tabelka} show that the conjecture is true for $k\leq 5$.

\medskip
Our next theorem gives a lower bound for $a_k$ which is ``very close'' to ${z_k}$.

\begin{theorem}\label{ograniczenia na a_k}
For every positive integer $k$,
\[
a_k\geq\sqrt{2k(k-1)}-k.
\]
\end{theorem}
\begin{proof} Clearly, the theorem holds for $k=1$, so we assume from now on that $k\geq 2$.

Consider a cyclic binary sequence $\boldsymbol{b} = b_0b_1\ldots b_{s-1}$ of length $s$ with minimum possible number $w_k(s)$ of $k$-bad pairs. For $i = 0,1,\ldots,s-1$ let $\overline{\ell}_i$ (resp. $\ell_i$) denote the number of $k$-bad (resp. $k$-good) pairs containing the term $b_i$. Clearly $\overline{\ell}_i + \ell_i = 2k$ for all $i$ and $w_k(s) = \frac{1}{2}\sum_{i=0}^{s-1} \overline{\ell}_i$. The following Claims \ref{obs1}-\ref{obs5} concern the sequence $\boldsymbol{b}$.

\begin{clm} \label{obs1}
For all $i$ we have $\overline{\ell}_i \leq k \leq \ell_i$.
\end{clm}
\begin{inproof}
Clearly, it suffices to show that $\ell_i \geq k$ because $\overline{\ell}_i  = 2k - \ell_i$. Suppose on the contrary that $\ell_i < k$ and assume without loss of generality that $b_i = 0$. Then, by switching the value of $b_i$ to $1$ we get a {cyclic} sequence with a smaller number of bad pairs which contradicts the minimality of the sequence $\boldsymbol{b}$.
\end{inproof}

Let $p_i$ denote the number of terms of $\boldsymbol{b}$ of value equal to $b_i$ which are at distance at most $k$  from both $b_i$ and $b_{i+1}$, excluding $b_i$ itself. Then
\begin{equation}\label{values of ell_i}
\overline{\ell}_i=p_i \quad {\rm or} \quad \overline{\ell}_i=p_i+1
\end{equation}
depending on the value of the term $b_{i-k}$.

\begin{clm} \label{obs2}
If $b_i = b_{i+1}$, then $\overline{\ell}_i -1 \leq \overline{\ell}_{i+1} \leq \overline{\ell}_i +1$ and $\ell_i -1 \leq \ell_{i+1} \leq \ell_i +1$.
\end{clm}

\begin{inproof}
As $b_i = b_{i+1}$, we have $\overline{\ell}_{i+1}=p_i$ or $\overline{\ell}_{i+1}=p_i+1$ depending on the value of the term $b_{i+1+k}$. Our statement follows now immediately by (\ref{values of ell_i}) and the equality $\overline{\ell}_i + \ell_i = 2k$.
\end{inproof}

\begin{clm} \label{obs3}
If $b_i \neq b_{i+1}$, then  $\ell_{i+1}-2 \leq \overline{\ell}_i \leq \ell_{i+1}$ and $\ell_i - 2 \leq \overline{\ell}_{i+1} \leq \ell_i$.
\end{clm}

\begin{inproof}
As $b_i \not= b_{i+1}$, $\ell_{i+1}=p_i+1$ or $\ell_{i+1}=p_i+2$ depending on the value of the term $b_{i+1+k}$ because $i(i+1)$ is a good pair. As before the statement follows by (\ref{values of ell_i}) and the equality $\overline{\ell}_i + \ell_i = 2k$. \end{inproof}

\medskip
Consider a maximal segment $\boldsymbol{b'}=b_{h}b_{h+1}\ldots b_{h+t-1}$ in {$\boldsymbol{b}$} consisting of terms of the same value (all $0$'s or all $1$'s). Clearly, the length $t$ of this segment satisfies the inequality $t\leq k+1$ because otherwise $\overline{\ell}_{h+1}>k$ contradicting Claim \ref{obs1}. The Claims \ref{obs4}-\ref{obs5} concern terms of any maximal segment $\boldsymbol{b'}$ of terms of the same value in $\boldsymbol{b}$.

\begin{clm} \label{obs4}
For every $j = 0,1,\ldots,t-1$, we have $\overline{\ell}_{h+j} \geq t-1$.
\end{clm}

\begin{inproof}
This observation follows from the fact that $t\leq k+1$ and the pairs $(h+i)(h+j)$, where $0\leq i< j\leq t-1$, are $k$-bad.
\end{inproof}

\begin{clm} \label{obs5}
For every $j = 0,1,\ldots,t-1$, we have $\overline{\ell}_{h+j} \geq k-2-j$ and $\overline{\ell}_{h+j} \geq k-t-1+j$.
\end{clm}

\begin{inproof}
Applying, in turn, Claims \ref{obs2}, \ref{obs3} and \ref{obs1} we get
\[
\overline{\ell}_{h+j} \geq \overline{\ell}_{h+j-1}-1\geq\ldots\geq \overline{\ell}_h-j \geq \ell_{h-1}-2-j \geq k-j-2.
\]
Similarly, using the same claims,
\[
\overline{\ell}_{h+j} \geq \overline{\ell}_{h+j+1}-1\geq\ldots\geq \overline{\ell}_{h+t-1}-t+1+j \geq \ell_{h+t}-2-t+1+j \geq k-t-1+j.
\]
\end{inproof}

For a maximal segment $\boldsymbol{b'}=b_{h}b_{h+1}\ldots b_{h+t-1}$ in $\boldsymbol{b}$ of terms of the same value we define the {\it score} $Sc(\boldsymbol{b'})=\frac{1}{t}\sum_{j=0}^{t-1} \overline{\ell}_{h+j}$. Clearly, by Claims \ref{obs4} and \ref{obs5},
\begin{multline}\label{ograniczenie gorne na dobre pary}
Sc(\boldsymbol{b'}) \geq \frac{1}{t}\sum_{j=0}^{t-1}\max(k-2-j,k-t-1+j,t-1) \\= k-2 - \frac{1}{t}\sum_{j=0}^{t-1}\min(j,t-1-j,k-t-1).
\end{multline}
We consider three cases.


\medskip
\noindent
{\textbf{Case 1:} $t < \frac{2k-1}{3}$}

\medskip
In this case $\min(j,t-1-j,k-t-1)$ reduces to $\min(j,t-1-j)$, so by (\ref{ograniczenie gorne na dobre pary}), we get
\begin{equation}\label{case1}
Sc(\boldsymbol{b'}) \geq k-2 - \frac{(t-1)^2}{4t}.
\end{equation}



\medskip
\noindent
{\textbf{Case 2:} $\frac{2k-1}{3}\leq t \leq k-1$}

\medskip
By (\ref{ograniczenie gorne na dobre pary}), we get
\begin{multline}\label{case2}
Sc(\boldsymbol{b'}) \geq k-2 - \frac{1}{t}[0+1+\ldots+(k-t-1) + (t-2(k-t))(k-t-1)\\ + (k-t-1)+(k-t-2)+\ldots+0]\\
=k-2 - \frac{1}{t}(k-t-1)(2t-k).
\end{multline}

\medskip
\noindent
{\textbf{Case 3:} $k\leq t\leq k+1$}

\medskip
In this case $\min(j,t-1-j,k-t-1)$ reduces to $k-t+1$, so by (\ref{ograniczenie gorne na dobre pary}), we get
\begin{equation}\label{case3}
Sc(\boldsymbol{b'}) \geq t-3.
\end{equation}

Let
\[
u(t)=\left\{ \begin{array}{ll}
   k-2 - \frac{(t-1)^2}{4t}  & \mbox{for $t<\frac{2k-1}{3}$} \\
   k-2 - \frac{1}{t}(k-t-1)(2t-k) & \mbox{for $\frac{2k-1}{3}\leq t \leq k-1$}\\
   t-3 & \mbox{for $k\leq t\leq k+1$}
 \end{array}
 \right.
\]
be a real valued function. It is routine to check that for $k\geq 2$, $u(t)$ reaches its minimum value in the interval $[1,k+1]$ at $t=\sqrt{\frac{k(k-1)}{2}}$.

Thus, by the inequalities (\ref{case1}), (\ref{case2}) and (\ref{case3})
\begin{equation}\label{score}
Sc(\boldsymbol{b'}) \geq u\left(\sqrt{\frac{k(k-1)}{2}}\right) = 2\sqrt{2k(k-1)}-2k.
\end{equation}

We divide $\boldsymbol{b}$ to maximal segments $\boldsymbol{b_1},\boldsymbol{b_2},\ldots,\boldsymbol{b_r}$ of terms of the same value. Let $t_i$ denote the length of $\boldsymbol{b_i}$. Clearly, the number of $k$-bad pairs in $\boldsymbol{b}$ is $\frac{1}{2}\sum_{i=0}^{s-1}\overline{\ell}_i=\frac{1}{2}\sum_{i=1}^r Sc(\boldsymbol{b_i})t_i$. Thus, by (\ref{score}),

\begin{align*}
w_k(s)=& \frac{1}{2}\sum_{i=1}^r Sc(\boldsymbol{b_i})t_i \geq \sum_{i=1}^r (\sqrt{2k(k-1)}-k)t_i = (\sqrt{2k(k-1)}-k)s.
\end{align*}
Hence, by (\ref{limit}), $a_k=\lim_{s\rightarrow\infty}\frac{w_k(s)}{s}\geq \sqrt{2k(k-1)}-k$.

\end{proof}

\bigskip
We have shown that
\begin{equation}\label{bounds}
\sqrt{2k(k-1)}-k\leq a_k\leq {z_k}.
\end{equation}
We shall see now that these bounds for $a_k$ are very close to each other.

First observe that there exists a positive integer $t$ such that
\begin{equation}\label{integer minimum}
t+\frac{(k+1)k}{2t}\leq \sqrt{2(k+1)k+1}.
\end{equation}
Indeed, consider the function $f(x)=x+\frac{(k+1)k}{2x}$ and let $x_1=\sqrt{\frac{(k+1)k}{2}+\frac{1}{4}}-\frac{1}{2}$. It is easy to verify that for $x\in [x_1,x_1+1]$, $f(x)\leq f(x_1)=f(x_1+1)=\sqrt{2(k+1)k+1}$, so we define $t$ to be the unique integer in the interval $[x_1,x_1+1)$.

By (\ref{def b_k}) and (\ref{integer minimum}) we get
\[
{z_k}=\min_{\frac{k}{2}\leq t\leq k+1} \left(t+\frac{(k+1)k}{2t}-k-1\right) \leq \sqrt{2(k+1)k+1}-k-1.
\]
Using the inequality above one can readily verify that the difference between the upper and the lower bound for $a_k$ given in (\ref{bounds}) is smaller than $0.5$ for $k\geq 5$ (and it tends to $\sqrt{2}-1$ as $k$ tends to infinity). Thus, since the actual value of $a_4$ differs from the lower bound in (\ref{bounds}) by less than $0.5$ too (see Table \ref{tabelka}), we have the following statement.
\begin{col}
\label{cor2}
For all integers $k\geq 4$,
\[
\sqrt{2k(k-1)}-k\leq a_k< \sqrt{2k(k-1)}-k+1/2.
\]
\end{col}\hfill$\Box$

\noindent
\begin{proof}[Proof of Theorem \ref{theorem_KmnUpper}.] Theorem \ref{theorem_KmnUpper} follows immediately from Theorem \ref{upper} and Corollary \ref{cor1} by defining $d_k=\frac{k}{k-a_k}$. The bounds for $d_k$ given in Theorem \ref{theorem_KmnUpper} can be easily obtained from Corollary \ref{cor2} for $k\geq 4$ and by direct computations (using Table \ref{tabelka}) for $k<4$.
\end{proof}

\begin{proof}[Proof of Corollary \ref{corollary_mcincirculant}]
Let $v_1,\ldots,v_n$ be the vertices of $C_n^k$, in a natural order. Any cyclic binary sequence ${\bf b}=b_1b_2\ldots b_n$ defines a bipartition $(V_0,V_1)$ of the vertex set of $C_n^k$: a vertex $v_i$ goes to $V_0$ if $b_i=0$ and it goes to $V_1$ otherwise. One can readily verify that the number of good pairs in ${\bf b}$ is equal to the size of the bipartition $(V_0,V_1)$, i.e. the number of edges joining vertices of the two sets $V_0$ and $V_1$. Consequently, the  size of a maximum cut in $C_n^k$ is equal to the maximum number $kn-w_k(n)$ of good pairs in a cyclic binary sequence of length $n$. Therefore, the proof is complete by the equalities (\ref{limit}) and $d_k=\frac{k}{k-a_k}$.
\end{proof}

\section{Complexity results for arbitrary graphs} \label{complexity}

In this section we consider problems of finding the numbers $f_k(G)$ and $c_k(G)$ for arbitrary connected graphs $G$.


Let us first make the definition of $c_k(G)$ a bit more precise. We define for a graph $G$ and $k<|V(G)|$ a {\it $k$-cover sequence} $\mathbf{c} = c_1,\ldots,c_m$ to be a sequence of $(k+1)$-subsets of $V(G)$ such that every two consecutive sets in $\mathbf{c}$ differ by one element (that is, $\left| c_i \setminus c_{i+1} \right| = \left| c_{i+1} \setminus c_{i} \right| = 1$ for $i=1, \ldots, m-1$) and for every edge $e\in E(G)$ we have $e\subseteq c_i$ for some $i$. Clearly, a $k$-cover sequence describes replacements of objects in the cache; if we assume that at time $0$ the cache holds the set $c_1$, then at time $t$ (for $1\leq t\leq m-1$) we replace the only object of $c_{t}\setminus c_{t+1}$ by the only object of $c_{t+1}\setminus c_{t}$. The number of read operations in this scenario is equal to $m+k$ (because we have to read the $k+1$ elements of $c_1$ at start). Consequently, $c_k(G)$ is equal to the sum of $k$ and the length of a shortest $k$-cover sequence for $G$.

We shall discuss the computational complexities of the following two {families of} decision problems.

\medskip
\begin{framed}
\noindent{\textbf{Problem:}} {\kradius{k}}\\
{\textbf{Instance:}} A connected graph $G$ and an integer $M$.\\
{\textbf{Question:}} Is there a $k$-radius sequence of length $M$ for the graph $G$?
\end{framed}

\medskip
\begin{framed}
\noindent{\textbf{Problem:}} {\kcover{k}}\\
{\textbf{Instance:}} A connected graph $G$ and an integer $M$.\\
{\textbf{Question:}} Is there a $k$-cover sequence of length $M$ for the graph $G$?
\end{framed}
\medskip

Let us consider the {\kradius{k}} problem first.

We start with a simple lower bound for the value of $f_k(G)$.

\begin{prop} \label{lower}
For any graph $G$, $f_k(G) \geq \sum_{v \in V} \left \lceil \frac{\deg v}{2k} \right \rceil$.
\end{prop}

\begin{proof}
Consider a shortest $k$-radius sequence $\bf x$ for $G$ and for each vertex $v$ let $m(v)$ denote the number of appearances of $v$ in $\bf x$. For every appearance of $v$ in $\bf x$, at most $2k$ neighbors of $v$ appear at distance at most $k$ in $\bf x$. Thus $m(v) \geq \left \lceil \frac{\deg v}{2k} \right \rceil$ and $f_k(G) = \sum_{v \in V} m(v) \geq \sum_{v \in V} \left \lceil \frac{\deg v}{2k} \right \rceil$.
\end{proof}

First consider the case $k=1$. The problem {\kradius{1}} asks for a sequence of vertices of $G$, in which the endvertices of every edge appear as consecutive elements.
Let $G$ be a connected graph with $n$ vertices and $m$ edges. Denote by $n_o$ the number of vertices of odd degree in $G$. We add $n_o/2$ edges between them, creating a multigraph $G'$, which has an Euler circuit. This Euler circuit corresponds to a 1-radius sequence in $G$ of length $ m + n_o/2$, which matches the lower bound given in Proposition \ref{lower}. Thus the problem {\kradius{1}} is polynomially solvable.

\begin{theorem}
For $k\geq 2$ the problem {\kradius{k}} is NP-complete.
\end{theorem}

\begin{proof} The problem is clearly in the class NP. To prove NP-hardness, we shall show a reduction from the problem of determining existence of a Hamiltonian path in a cubic triangle-free graph. Determining the existence of a Hamiltonian path in planar, cubic, 3-connected graphs with every face bounded by at least 5 edges is NP-complete, see Garey and Johnson \cite[p. 199]{GJ}. These graphs are triangle-free; if there were a triangle $uvw$ then, by $3$-connectivity and $3$-regularity, the remaining neighbors of $u,v$ and $w$ would be either all inside or all outside the triangle, so $uvw$ would be a face bounded by $3$ edges, a contradiction.

Let $F$ be a cubic triangle-free graph with $n$ vertices. For every vertex $v$ in $F$ we add $k-2$ pendant edges  $e_v^1,\ldots,e_v^{k-2}$ incident with $v$ and call the resulting graph $F'$. We define $G$ to be the line graph of $F'$. We shall prove that there exists a $k$-radius sequence for $G$ of length $\frac{2}{k+1}e(G)+1$ if and only if $F$ has a Hamiltonian path.

Assume first that $F$ has a Hamiltonian path $H$ whose consecutive vertices are $v_1,v_2,\ldots,v_n$. For $i=1,2,\ldots,n-1$, let $e_i$ be the edge $v_iv_{i+1}$. For $i=2,3,\ldots,n-1$, we define $f_i$ to be the edge in $F$ incident with $v_i$ which is not in $H$. Moreover, let $e_0,f_1$ (resp. $f_n$, $e_n$) be the edges incident in $F$ with $v_1$ (resp. $v_n$) which are not in $H$ (it is possible that $e_0=e_n$). Clearly, each edge of $F$ which is not an edge of $H$ appears twice in the sequence $e_0,f_1,f_2,\ldots,f_n,e_n$ (see Figure \ref{reduction-rk}).

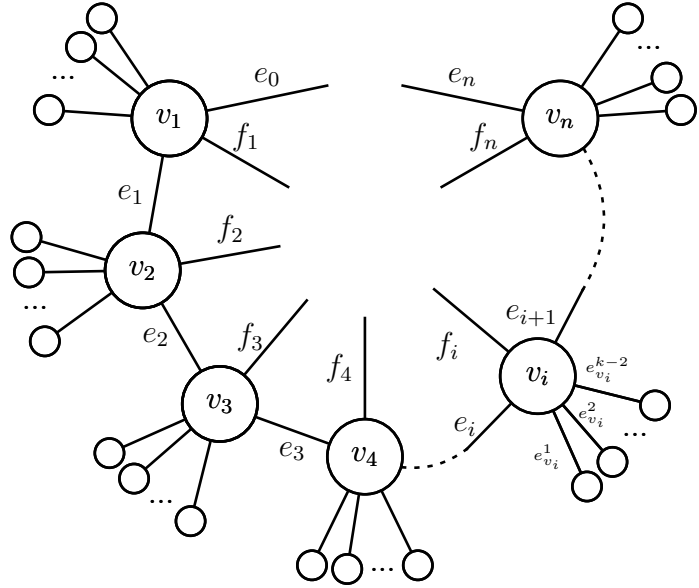
\begin{figure}
\centering
\begin{tikzpicture}[line cap=rect,line width=1pt,scale=1.0]

\foreach \i/\ang in {1/-210,2/-170,3/-130,4/-90,i/-40,n/30}
{
\node[draw, circle, minimum size = 1cm] (v\i) at (\ang:3cm) {$v_{\i}$};
\node (y\i) at (\ang:1cm) {};

\node[draw, circle] (x\i1) at (\ang-9:4.5cm) {};
\node[draw, circle] (x\i2) at (\ang-3:4.5cm) {};
\node (x\i3) at (\ang+3:4.5cm) {...};
\node[draw, circle] (x\i4) at (\ang+9:4.5cm) {};

\node[draw, circle, minimum size = 1cm] (v\i) at (\ang:3cm) {$v_{\i}$};
\node (y\i) at (\ang:1cm) {};

\node[draw, circle] (x\i1) at (\ang-9:4.5cm) {};
\node[draw, circle] (x\i2) at (\ang-3:4.5cm) {};
\node (x\i3) at (\ang+3:4.5cm) {...};
\node[draw, circle] (x\i4) at (\ang+9:4.5cm) {};

\draw[-] (v\i) --  (x\i1);
\draw[-] (v\i) --  (x\i2);
\draw[-] (v\i) --  (x\i4);
}
\node (y11) at (-260:2cm) {};
\node (ynn) at (80:2cm) {};

\draw[-] (v1) -- node[left]{$e_1$} (v2);
\draw[-] (v2) -- node[left]{$e_2$} (v3);
\draw[-] (v3) -- node[below]{$e_3$} (v4);

\coordinate (vip) at (-15:3cm);
\coordinate (vim) at (-65:3.2cm);
\draw[loosely dotted] (v4) to[bend right=20]  (vim);
\draw[-] (vim) to node[above, left]{$e_i$}  (vi);
\draw[-] (vi) to node[above, left] {$e_{i+1}$}  (vip);
\draw[loosely dotted] (vip) to[bend right=30]  (vn);

\draw[-] (v1) -- node[above]{$e_0$} (y11);
\draw[-] (v1) -- node[above]{$f_1$} (y1);
\draw[-] (v2) -- node[above]{$f_2$} (y2);
\draw[-] (v3) -- node[left]{$f_3$} (y3);
\draw[-] (v4) -- node[left]{$f_4$} (y4);
\draw[-] (vi) -- node[below left]{$f_i$} (yi);
\draw[-] (vn) -- node[above]{$f_n$} (yn);
\draw[-] (vn) -- node[above]{$e_n$} (ynn);

\node[draw, circle, minimum size = 1cm] (v11) at (-210:3cm) {};
\node[draw, circle, minimum size = 1cm] (v12) at (-170:3cm) {};
\node[draw, circle, minimum size = 1cm] (v13) at (-130:3cm) {};

\draw[-] (vi) --node[near end, above, left]{\tiny $e^1_{v_i}$}  (xi1);
\draw[-] (vi) --node[near end, above]{\tiny $e^2_{v_i}$}  (xi2);
\draw[-] (vi) --node[above]{\tiny $e^{k-2}_{v_i}$}  (xi4);
\end{tikzpicture}

\caption{The reduction from the Hamiltonian path problem in cubic triangle-free graphs to {\kradius{k}}.}
\label{reduction-rk}
\end{figure}

We claim that
\[
e_0,f_1,e^1_{v_1},\dots,e^{k-2}_{v_1},e_1,f_2,e^1_{v_2},\dots,e^{k-2}_{v_2}e_2,\ldots,e_{n-1},f_n,e^1_{v_n},\dots,e^{k-2}_{v_n}e_n
 \]
is the required $k$-radius sequence for the graph $G$. As $G$ has $\frac{(k+1)kn}{2}$ edges, the length of the sequence is $kn+1=\frac{2}{k+1}e(G)+1$. The set of edges of $G$ is
\begin{align*}
\{ &e_{i-1}e_i,e_{i-1}f_i,e_if_i:\ 1\leq i \leq n\}\\
& \cup \{ e_{v_i}^je_{i-1},e_{v_i}^je_i,e_{v_i}^jf_i:\ 1\leq i\leq n \ {\rm and } \ 1\leq j\leq k-2\}\\
& \cup \{ e_{v_i}^je_{v_i}^\ell:\ 1\leq i\leq n \ {\rm and }\ 1\leq j<\ell \leq k-2\},
\end{align*}
so it is clear that the sequence defined above is a $k$-radius sequence for $G$.

\medskip
To show the converse we assume now that there is a $k$-radius sequence ${\bf x}$ for $G$ of length $\frac{2}{k+1}e(G)+1=kn+1$. Let $E_F$ be the set of vertices in $G$ which are edges in $F$.

We shall prove the following statement first.

\begin{clm} \label{NPclaim}
If two vertices $a,b\in E_F$ in $G$, $a\not=b$, appear exactly once in the sequence ${\bf x}$ and they are at distance at most $k$ in this sequence, then they are at distance exactly $k $ in ${\bf x}$.
\end{clm}

\begin{inproof}
Clearly, the degree of $a$ in $G$ is equal to $2k$, so as $a$ appears only once in ${\bf x}$, all vertices at distance at most $k$ in ${\bf x}$ are neighbors of $a$ in $G$. In particular $b$ is a neighbor of $a$. Thus the edges $a$ and $b$ in $F$ have a common endvertex, say $v$. Let $c$ be the third edge in $F$ incident with $v$.

Observe that $a$ and $b$ have no common neighbors in $G$ except $c,e_v^1,e_v^2,\ldots,e_v^{k-2}$. Indeed, suppose $d$ is a neighbor of both $a$ and $b$ in $G$ different from $c,e_v^1,e_v^2,\ldots,e_v^{k-2}$. By the definition of $F'$, $d$ is an edge in $F$. Moreover, $v$ is not an  endvertex of $d$ in $F$ because the degree of $v$ in $F$ is equal to $3$. Thus, the edges $a$, $b$ and $d$ in $F$ form a triangle, a contradiction.

Suppose some neighbor, say $g$, of $a$ in $G$, which is not a neighbor of $b$ appears at distance at most $k$ in ${\bf x}$ from both $a$ and $b$. Then, as $b$ has $2k$ neighbors in $G$ and it appears only once in ${\bf x}$, some neighbor of $b$ is not at distance at most $k$ from $b$ in ${\bf x}$, a contradiction. Thus, the positions in ${\bf x}$ between the occurrence of $a$ and the occurrence of $b$ must be occupied by the $k-1$ common neighbors $c,e_v^1,e_v^2,\ldots,e_v^{k-2}$ of both $a$ and $b$, which completes the proof of the claim.
\end{inproof}

Let us denote by $A_F$ the set vertices in $E_F\subseteq V(G)$ which appear exactly once in ${\bf x}$ and let $m=|A_F|$. It follows from Claim \ref{NPclaim} that $m\leq n-1$. To see this, observe first that any $a\in A_F$ can appear neither at the first $k$ nor at the last $k$ positions in the sequence ${\bf x}$ because $a$ has $2k$ neighbors in $G$ so there must be $2k$ positions in ${\bf x}$ at distance at most $k$ from the occurrence $a$. Thus, the sequence ${\bf x}$ has $kn+1-2k=(n-2)k+1$ available positions for the members of $A_F$. If $m>n-1$, then there are two members of $A_F$ that are at distance strictly smaller than $k$ in ${\bf x}$ which contradicts Claim \ref{NPclaim}. We have shown that $m\leq n-1$.

Since $\frac{3n}{2}-m$ members of $E_F$ appear at least twice in the sequence ${\bf x}$ (of length $kn+1$) and there are $(k-2)n$ vertices of $G$ which are not in $E_F$,
\[
kn+1\geq m + (k-2)n +2\left(\frac{3n}{2}-m\right) = (k+1)n-m,
\]
so $m\geq n-1$.

We have shown that $|A_F|=m=n-1$. By Claim \ref{NPclaim}, consecutive appearances of vertices of $A_F$ in the sequence ${\bf x}$ are at distance exactly $k$. In other words we proved that there is an ordering of the vertices $a_1,a_2,\ldots,a_{n-1}$ in $A_F$ such that, for $i=1,2,\ldots,n-1$, $a_i$ is at distance equal to $k$ from $a_{i+1}$ in the sequence ${\bf x}$.

As all neighbors in $G$ of each vertex  $a_i\in A_F$ are at distance at most $k$ in ${\bf x}$, the vertices $a_1,a_2,\ldots,a_{n-1}$ form an induced path in the graph $G$. Recall that $G$ is a line graph of $F'$ and $A_F\subseteq E(F)\subseteq E(F')$, so $a_1,a_2,\ldots,a_{n-1}$ are edges of a Hamiltonian path in $F$.
\end{proof}

We proceed to the problem {\kcover{k}}.

{
\begin{theorem}
For all $k\geq 1$ the problem {\kcover{k}} is NP-complete.
\end{theorem}

\begin{proof}
Clearly the problem is in NP. We shall first prove NP-hardness for $k=1$ and then for all $k \geq 2$.

Consider the case $k=1$. Let $H$ be a graph with $m$ edges and by $L(H)$ we denote the line graph of $H$.
Clearly, a 1-cover sequence for $H$ has to have at least $m$ terms. Moreover, $H$ has a 1-cover sequence of length exactly $m$ if and only if $L(H)$ has a Hamiltonian path (this path corresponds to the shortest 1-cover sequence). Since determining the existence of a Hamiltonian path in line graphs is NP-complete \cite{Be}, the problem {\kcover{1}} is NP-complete as well.

Now consider the problem {\kcover{k}} for $k \geq 2$. We shall show a reduction from {\kcover{1}}.
Let $H$ be a graph with $m$ edges and let $N = \binom{k}{2}(m-1)+\binom{k+1}{2}+3$. We define a graph $G$ in which every edge $uv$ of $H$ is replaced by the ,,edge gadget'' $G_{uv}$, depicted in Figure \ref{reduction-bk}.

We introduce the set $K_{uv} = \{x_{uv}^1,x_{uv}^2,\ldots,x_{uv}^k\}$ of $k$ vertices forming a clique. We also add $N-2$ new vertices $y_{uv}^1,y_{uv}^2,\ldots,y_{uv}^{N-2}$. Finally, we add an edge between every vertex from $K_{uv}$ and every vertex from $\{u,v, y_{uv}^1,y_{uv}^2,\ldots,y_{uv}^{N-2}\}$.

Clearly, the graph $G$ has $m \left({k\choose 2}+Nk \right)$ edges.
{
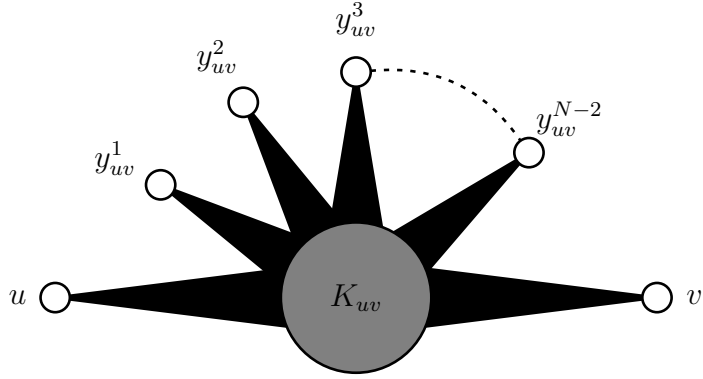
\begin{figure}[h]
\centering
\begin{tikzpicture}[line cap=rect,line width=1pt,scale=1]

\coordinate (u) at (-4,0);
\coordinate (v) at (4,0);
\coordinate (Kuv) at (0,0);

\foreach \i/\ang in {1/150,2/120,3/90,N-2/40}
{
	\draw[fill = black] (\ang:3cm) -- (\ang-20:1cm) -- (\ang+20:1cm)  -- cycle;
	\node[draw, circle,fill=white] (y\i) at (\ang:3cm) {};
	\node at (\ang:3.7cm) {$y_{uv}^{\i}$};
}
\draw[loosely dotted] (y3) to[bend left]  (yN-2);
\draw[fill = black] (u) -- (90:0.5cm) -- (-90:0.5cm)  -- cycle;
\draw[fill = black] (v) -- (90:0.5cm) -- (-90:0.5cm)  -- cycle;
\node[draw, circle, minimum size = 2 cm, fill=gray] at (Kuv) {$K_{uv}$};
\node[draw, circle,fill=white] at (u) {};
\node at (-4.5,0) {$u$};
\node[draw, circle,fill=white] at (v) {};
\node at (4.5,0) {$v$};
\end{tikzpicture}

\caption{The ,,edge gadget'' $G_{uv}$ used in the reduction from {\kcover{1}} to {\kcover{k}} (for $k\geq 2$).}
\label{reduction-bk}
\end{figure}
}

We shall prove that there exists a $k$-cover sequence of length $mN + (m-1)(k-1)$ for the graph $G$ if and only if there exists a 1-cover sequence of length $m$ for the graph $H$.

Assume first that $H$ has a 1-cover sequence $\bf x$ of length $m$. As we observed before, it is a sequence of edges of $H$, such that every edge of $H$ appears exactly once in $\bf x$ and consecutive edges are adjacent. Let ${\bf x} = e_1,e_2,e_3,\ldots,e_m$.

Consider an edge $uv$ of $H$ and the following sequence of $(k+1)$-element sets
\[
{\bf c}_{uv}=K_{uv}\cup\{ u \},K_{uv}\cup\{ y_{uv}^1\},K_{uv}\cup\{ y_{uv}^2\},\ldots,K_{uv}\cup\{ y_{uv}^{N-2}\},K_{uv}\cup\{ v \}.
\]
Obviously, the sequence ${\bf c}_{uv}$ is a $k$-cover sequence for the graph $G_{uv}$ of length $N$. Unfortunately, the concatenation of the sequences ${\bf c}_{uv}$ (consistent with the order of edges in $\bf x$) is not a $k$-cover sequence for the whole graph $G$ because for consecutive edges $uv$ and $vw$ in $\bf x$, the last term of ${\bf c}_{uv}$ differs by more than one element from the first term of ${\bf c}_{vw}$. Therefore, we have to use some auxiliary ,,connector'' sequences, each of length $k-1$:
\begin{multline*}
{\bf d}_{uv \to vw}= (K_{uv} \setminus \{x_{uv}^1\}\cup \{v,x^1_{vw}\}),(K_{uv} \setminus \{x_{uv}^1,x_{uv}^2\}\cup \{v,x^1_{vw},x^2_{vw}\}),\ldots,\\
(\{x_{uv}^{k}\} \cup ((K_{vw} \setminus \{x_{vw}^{k}\})\cup\{v\})).
\end{multline*}
It is straightforward to check that the sequence
\[
{\bf c}_{e_1},{\bf d}_{e_1 \to e_2},{\bf c}_{e_2},{\bf d}_{e_2 \to e_3},\ldots,{\bf d}_{e_{m-1}\to e_m},{\bf c}_{e_m}
\]
is a $k$-cover sequence for the graph $G$. The length of this sequence is $mN + (m-1)(k-1)$, as claimed.

Before showing the converse implication we need to give a few definitions. Let ${\bf c}=c_1,\ldots,c_s$ be a $k$-cover sequence for the graph $G$. We will say, that a vertex $v\in V(G)$ is {\it loaded} (or {\it loaded into cache}) at step $i>1$ if $c_i \setminus c_{i-1}=\lbrace v\rbrace $. An edge $vu\in E(G)$ is {\it covered} at step $i$ if $\lbrace v,u\rbrace \subseteq c_{i}$ and, we loaded $v$ or $u$ at step $i$.

We say, that a pair $uv$ of vertices of $G$ is a {\it loss} at step $i>1$, if $\lbrace u,v\rbrace\subseteq c_i$, $\lbrace u,v\rbrace\not\subseteq c_{i-1}$ and either (a) $uv\notin E(G)$ or (b) $\lbrace u,v\rbrace \subseteq c_{j}$, for some $j<i-1$. (Intuitively, each time we load to our cache a vertex which does not form an edge in $G$ with some vertex already in the cache or forms an edge that has been covered before, we report a loss.) Moreover, a pair $uv$ is a loss at step $1$, if $\lbrace u,v\rbrace\subseteq c_1$ and $uv\notin E(G)$. Let $\ell({\bf c})$ be the sum of the numbers of losses at steps $1,2,\ldots,m$.

Note, that it follows from the definitions above that if ${\bf c}$ is a $k$-cover sequence for $G$, then
\begin{equation}\label{loss}
e(G)+\ell({\bf c})=k(s-1)+{k+1\choose 2}.
\end{equation}

Assume now that ${\bf c}$ is a $k$-cover sequence for $G$ of length $s=mN + (m-1)(k-1)$. Since $e(G)=m\left ( \binom{k}{2} + Nk \right)$, the equality (\ref{loss}) implies that
\begin{equation}\label{number of losses}
\ell({\bf c})={k\choose 2}(m-1).
\end{equation}

We shall show two claims.
\begin{clm}\label{claim_N}
For every $e \in E(H)$, there is a term $c_j$ of ${\bf c}$ such that $c_j\subseteq V(G_e)$.
\end{clm}
\begin{inproof}
Suppose no term of ${\bf c}$ is contained in $V(G_e)$. Then, at each step when any of the vertices $y^1_e,\ldots,y^{N-2}_{e}$ is loaded, we generate a loss (because all neighbors of each $y^i_e$ are in $V(G_e)$). Thus, the total number of losses $\ell({\bf c})\geq N-2-{k+1\choose 2}>{k\choose 2}(m-1)$, a contradiction with the equality (\ref{number of losses}).
\end{inproof}

\begin{clm}\label{claim_transition}
For $e,f \in E(H)$, $e \neq f$, let $c_i$ and $c_j$, $i<j$, be terms of ${\bf c}$ contained in $V(G_e)$ and $V(G_f)$, respectively. Then the total number of losses at steps $i+1,i+2,\ldots,j$ is at least ${k\choose 2}$.
Moreover, if $e \cap f = \emptyset$, then the total number of losses in these steps is larger than ${k\choose 2}$.
\end{clm}
\begin{inproof}
By the construction of $G$, $|V(G_e)\cap V(G_f)|\leq 1$ and $V(G_e)\cap V(G_f)\neq\emptyset$ if and only if $e \cap f \neq \emptyset$. Thus, there are at least $k$ different vertices which are not in $V(G_e)$ that are loaded at steps $i+1,i+2,\ldots,j$.

Let $w_1,\ldots,w_{k}\not\in V(G_e)$ be such vertices which are loaded first (i.e. $w_1\not\in V(G_e)$ is loaded first in steps $i+1,i+2,\ldots,j$, $w_2\not\in V(G_e)$ is loaded second, etc.).

Observe that when we load some vertex $w_t$, $1\leq t\leq k$, into the cache, at least $k-t+1$ vertices of the cache are members of $V(G_e)$. By the construction of $G$ again, vertices which are not in $V(G_e)$ have at most one neighbor in $V(G_e)$ (otherwise $H$ has multiple edges). Hence, when we load $w_t$, we generate at least $k-t$ losses. Consequently, the total number of losses when we load  $w_1,\ldots,w_{k}$ is at least $\sum_{t=1}^{k}(k-t)={k\choose 2}$ which proves the first part of the claim.

To show the second part suppose on the contrary that $e \cap f = \emptyset$ and we have exactly ${k\choose 2}$ losses at steps $i+1,i+2,\ldots,j$. Let $e = uv$ and $f=u'v'$. Then, when we load each $w_t$, we generate exactly $k-t$ losses. In particular, loading $w_1$ generates exactly $k-1$ losses. This can happen only if $w_1$ has a neighbor in $V(G_e)$. By construction of $G$, this neighbor can be either $u$ or $v$ -- without loss of generality assume it is $v$. Moreover, we know that $w_1 \in K_{e'} \subseteq G_{e'}$, for some $e' \in E(H)$ such that $v \in e'$. Clearly, the vertex $v$ is the only neighbor of $w_1$ in $V(G_e)$. In general, when we load any vertex $w_t$, to generate no more than $k-t$ losses, the vertices $v,w_1,\ldots,w_{t-1}$ must be neighbors of $w_t$. But the only common neighbors of $v$ and $w_1$ are members of $K_{e'}$, so $w_1,\ldots,w_{k}\in K_{e'} \subseteq V(G_{e'})$.

We have shown that $c_s=\{ v,w_1,\ldots,w_{k}\}\subseteq K_{e'}$, for some $s$, where $i+1\leq s \leq j$. Clearly, $e' \not= f$ because $v \in e \cap e'$, while $e \cap f = \emptyset$. Thus, by the first part of the claim there are at least ${k\choose 2}\geq 1$ losses at steps $s+1,s+2,\ldots,j$ so the total number of losses at steps $i+1,i+2,\ldots,j$ is larger than ${k\choose 2}$. This contradiction completes the proof of the claim.
\end{inproof}

\bigskip
By Claim \ref{claim_N}, there is an ordering $e_1,e_2,\ldots,e_m$ of the edges of $H$ and indices $j_1<j_2<\ldots<j_m$ such that $c_{j_i}\subseteq V(G_{e_i})$, for $i=1,\ldots,n$. We observe that $e_1,e_2,\ldots,e_m$ are consecutive vertices of a 1-cover sequence for $H$. Indeed, if $e_i \cap e_{i+1} = \emptyset$, for some $i$, then by Claim \ref{claim_transition}, $\ell({\bf c})>{k\choose 2}(m-1)$ which contradicts the equality (\ref{number of losses}).
\end{proof}
}

\section{Concluding remarks and open problems}

Theorem \ref{theorem_KmnUpper} gives a good estimate of $f_k(K_{m,n})$ when both $m$ and $n$ are large (tending to infinity), which leaves the following question open: what happens when only one of these parameters goes to infinity? This problem probably cannot be solved using the same proof technique (in particular, the assumption (2) of Frankl-R\"odl theorem will not be satisfied) and it is not clear what the precise answer should be.

\begin{problem}
What is an asymptotically tight estimate on $f_k(K_{m,n})$ when $m$ is constant and $n$ tends to infinity?
\end{problem}

We would like to see an analog of Theorem \ref{theorem_KmnUpper} for other classes of graphs -- in particular, for complete $t$-partite graphs $K_{n_1,n_2,\ldots,n_t}$
(for fixed $t$ and large $n_1,n_2,\dots,n_t$). Proofs of Proposition \ref{norm}, Corollary \ref{cor1} and Theorem \ref{upper} can be adapted to the $t$-partite case, but the main difficulty is determining (an analog of) the constant $a_k$.

Since this problem seems to be interesting on its own, we will make it more precise. An unordered pair $ij$, where $i\neq j$, is a {\it $k$-bad} pair in a $t$-ary sequence ${\bf r}$ if $\left|j-i\right|\leq k$ and $r_i=r_j$. Let $w_{k,t}(s)$ be the minimum number of $k$-bad pairs in a $t$-ary sequence of length $s$ and set $a_{k,t}=\lim_{s\rightarrow \infty}\frac{w_{k,t}(s)}{s}$. The existence of this limit follows from an analog of Proposition \ref{norm} for $t$-ary sequences.

\begin{problem}
What is a nontrivial estimate on $a_{k,t}$?
\end{problem}

Note that Corollary \ref{cor2} gives a value of $a_{k,2}$ that is accurate only up to $\frac{1}{2}$, which means that we do not know the exact value of the constant $d_k$ in Theorem \ref{theorem_KmnUpper}. We know the values of $a_{k,2}$ for $k\leq 5$ (see Table \ref{tabelka}) and it would be interesting to find a precise formula for $k>5$.

Another interesting direction is to investigate the ratio of $f_k(G)$ and $c_k(G)$ for various graphs $G$ and fixed $k$. Recall that $\frac{f_k(K_n)}{c_k(K_n)}\to 1$ for $n\to \infty$ and $\frac{f_k(K_{n,m})}{c_k(K_{n,m})}\to d_k$ for $n,m\to\infty$ (and $d_k\approx 1.7$ for large $k$). It would be nice to know how large this ratio can be for other graphs. In particular the following question seems to be interesting.


\begin{problem}
What is the asymptotic behavior of the function
\[
g_k(n)=\max\left\{\frac{f_k(G)}{c_k(G)}:\ G\ {\rm has}\ n\ {\rm vertices}\right\},
\]
for a fixed $k$?
\end{problem}

It is not hard to show that $\lim_{n\rightarrow \infty}g_1(n)=\frac{3}{2}$ but we do not know anything about the behavior of $g_k(n)$ for $k>1$.

\section{Acknowledgments}
We are grateful to anonymous reviewers for their remarks that allowed us to improve the presentation of the results included in this paper.

\end{document}